\documentclass[letterpaper, 10 pt, conference]{ieeeconf}  % Comment this line out
\IEEEoverridecommandlockouts                              % This command is only
\usepackage{graphics} % for pdf, bitmapped graphics files
\usepackage{epsfig} % for postscript graphics files
\usepackage{amsmath} % assumes amsmath package installed
\usepackage{amssymb}  % assumes amsmath package installed
\usepackage{accents}
\usepackage{amsfonts}
\usepackage{mathtools}
\usepackage{enumerate}
\usepackage{tikz}
\usepackage[font = small]{caption}
\usepackage{subcaption}
\usepackage[noadjust]{cite} % basic fos bibliography [1]-[5]
\usepackage{verbatim}
\usepackage{color}
\usepackage[ruled,noend,boxed,vlined]{algorithm2e}
\usetikzlibrary{calc,positioning,shapes,shadows,arrows,fit}

\newtheorem{problem}{Problem}
\newtheorem{definition}{Definition}
\newtheorem{theorem}{Theorem}
\newtheorem{assumption}{Assumption}

\title{\LARGE \bf
Mode Switching Decentralized Multi-Agent Coordination under Local Temporal Logic Tasks
}

\author{Christos K. Verginis and Dimos V. Dimarogonas% <-this % stops a space
\thanks{The authors are with the KTH Center of Autonomous Systems, School of Electrical Engineering and Computer Science, KTH Royal Institute of Technology, SE-100 44, Stockholm, Sweden. Email: {\tt\small \{cverginis, dimos\}@kth.se}. This work was supported by the H2020 ERC Starting Grant BUCOPHSYS, the European Union's Horizon 2020 Research and Innovation Programme under the GA No. 731869 (Co4Robots), the Swedish Research Council (VR), the Knut och Alice Wallenberg Foundation (KAW) and the Swedish Foundation for Strategic Research (SSF).% <-this % stops a space
}}

\begin{document}

\maketitle
\thispagestyle{empty}
\pagestyle{empty}

%%%%%%%%%%%%%%%%%%%%%%%%%%%%%%%%%%%%%%%%%%%%%%%%%%%%%%%%%%%%%%%%%%%%%%%%%%%%%%%%
\begin{abstract}
This paper presents a novel control strategy for the coordination of a multi-agent system subject to high-level goals expressed as linear temporal logic formulas. In particular, each agent, which is modeled as a sphere with $2$nd order dynamics, has to satisfy a given local temporal logic specification subject to connectivity maintenance and inter-agent collision avoidance. We propose a novel continuous control protocol that guarantees navigation of one agent to a goal point, up to a set of collision-free initial configurations, while maintaining connectivity of the initial neighboring set and avoiding inter-agent collisions. Based on that, we develop a hybrid switching control strategy that ensures that each agent satisfies its temporal logic task. Simulation results depict the validity of the proposed scheme.

\end{abstract}

%%%%%%%%%%%%%%%%%%%%%%%%%%%%%%%%%%%%%%%%%%%%%%%%%%%%%%%%%%%%%%%%%%%%%%%%%%%%%%%%
\section{INTRODUCTION}\label{sec:Introduction}

The integration of temporal logic planning and multi-agent control systems has gained significant amount of attention during the last decade, since it provides planning capabilities that allow achievement of complex goals (see e.g.,  \cite{Chen2012,Diaz2015,Fainekos2009,Filippidis2012,Loizou2004,Meng15,Tumova2014,Zhang2016,verginis2017distributed,verginis2017robust,guo2016communication,kloetzer2011multi,nikou17TimedTemporal,saha2014automated,ulusoy2013optimality,verginis18TASE}). Firstly, an abstracted discrete version (e.g., a transition system) of the multi-agent system is derived by appropriately discretizing the workspace and finding the control inputs that navigate the system among the discrete states. 
A task specification is then given as a temporal logic formula (e.g., linear temporal logic (LTL) or metric-interval temporal logic (MITL)) with respect to the discretized version of the system, and by employing formal verification techniques, a high-level discrete plan is found that satisfies the task. Finally, the control inputs associated with the transitions between the discrete states are applied to achieve the plan execution.

An appropriate abstraction of the continuous-time system to a transition system form necessitates the design of appropriate control inputs for the transition of the system among the discrete states.  Most works in the related literature, when designing such discrete representations, either assume that there exist such control inputs or adopt simplified dynamics and employ optimization and input discretization techniques. Moreover, when deploying multi-robot teams, it is crucial to guarantee inter-agent collision avoidance. The latter is usually not taken into account in the related works, most of which unrealistically consider point-mass agents. 

Collision avoidance properties during the multi-agent transitions are incorporated in \cite{Loizou2004} and \cite{Filippidis2012}, where the authors adopt single-integrator models and appropriately constructed potential fields, namely navigation functions. These results, however, are not extendable to higher order dynamics in a straightforward way and are based on gain tuning, which might be problematic for real robot actuators. They also guarantee the multi-agent transitions from \textit{almost} all (except for a set of measure zero) collision-free initial conditions, implying that there are initial configurations that drive the multi-agent system to local minima. Potential-based collision avoidance was also incorporated in our previous works \cite{verginis18TASE}, where a centralized controller was employed, and \cite{verginis2017robust}, where no explicit potential field was given. In the latter, the agents also start their transitions simultaneously, which induces a centralized feature to the scheme.

In this paper, we propose a novel hybrid control strategy for the coordination of a multi-agent system subject to complex specifications expressed as linear temporal logic (LTL) formulas over predefined points of interest in the workspace. We first use formal verification methodologies to derive a high-level navigation plan for each agent over these points that satisfies its LTL formula. Then, we design a continuous control protocol that guarantees the global navigation of an agent to a goal point while guaranteeing inter-agent collisions and connectivity maintenance of the initially connected agents. By ``global", we mean up to a set of collision-free and connected (in the sense of a connected communication graph) initial configurations. The control scheme is decentralized, based on limited sensing capabilities of the agents, as well as robust to modeling uncertainties. Finally, by introducing certain priority variables for the agents, we develop a switching protocol that guarantees the sequential navigation of the agents to their goal points of interest and the satisfaction of their respective formulas.

This work can be considered as an extension of  \cite{guo2016communication}, where a similar strategy is followed.  In \cite{guo2016communication}, however, point-mass agents are considered and no inter-agent collision avoidance is taken into account. Moreover, the multi-agent transitions are not guaranteed globally; appropriate gain tuning achieves transitions from almost all initial conditions, except for a set of measure zero.

%Here I will write that most previous works do either abstraction analysis instead of designing a controller to define abstraction or discretize the control input. With respect to Meng's work, I will emphasize that here we do not have point robots and we need to ensure collision avoidance apart from connectivity maintenance. Also our controller works from \textit{all} initial conditions (collision-free and connected) and is much simpler than Meng's potential field, which is \textit{almost-all-initial-conditions}-based with gain tuning.

The rest of the paper is organized as follows. Section \ref{sec:Notation and Preliminaries} introduces notation and preliminary background. Section \ref{sec:problem form} provides the problem formulation and Section \ref{sec:main results} discusses the proposed solution. Simulation results are given in Section \ref{sec:Simulation} and Section \ref{sec:Conclusion} concludes the paper.

\section{Notation and Preliminaries} \label{sec:Notation and Preliminaries}

\subsection{Notation} \label{subsec:Notation}
The set of natural and real numbers is denoted by $\mathbb{N}$, and $\mathbb{R}$, respectively, and  $\mathbb{R}_{\geq 0}$, $\mathbb{R}_{> 0}$ are the sets of nonnegative and positive real numbers, respectively. The notation $\|x\|$ implies the Euclidean norm of a vector $x\in\mathbb{R}^n$. The  identity matrix is $I_n\in\mathbb{R}^{n \times n}$ and, given a  sequence $s_1\dots s_n$ of elements in $S$, we denote by $(s_1\dots s_n)^\mathsf{\omega}$ the infinite sequence $s_1\dots s_n s_1\dots s_n\dots $ created by repeating $s_1\dots s_n$.

\subsection{Task Specification in LTL} \label{subsec:LTL}

We focus on the task specification $\phi$ given as a Linear Temporal Logic (LTL) formula. The basic ingredients of a LTL formula are a set of atomic propositions $\Psi$ and several boolean and temporal operators. LTL formulas are formed according to the following grammar \cite{baier2008principles}: $\phi ::= \mathsf{true}\: |\:a\: |\: \phi_{1} \land  \phi_{2}\: |\: \neg \phi\: |\:\bigcirc \phi\:|\:\phi_{1}\cup\phi_{2} $, where $a\in \Psi$, $\phi_1$ and $\phi_2$ are LTL formulas and $\bigcirc$, $\cup$ are the \textit{next} and \textit{until} operators, respectively. Definitions of other useful operators like $\square$ (\it always\rm), $\lozenge$ (\it eventually\rm) and $\Rightarrow$ (\it implication\rm) are omitted and can be found at \cite{baier2008principles}.
The semantics of LTL are defined over infinite words over $2^{\Psi}$. Intuitively, an atomic proposition $\psi\in \Psi$ is satisfied on a word $w=w_1w_2\dots$, denoted by $w\models\psi$, if it holds at its first position $w_1$, i.e. $\psi\in w_1$. Formula $\bigcirc\phi$ holds true if $\phi$ is satisfied on the word suffix that begins in the next position $w_2$, whereas $\phi_1\cup\phi_2$ states that $\phi_1$ has to be true until $\phi_2$ becomes true. Finally, $\lozenge\phi$ and  $\square\phi$ holds on $w$ eventually and always, respectively. For a full definition of the LTL semantics, the reader is referred to \cite{baier2008principles}.

\section{Problem Formulation} \label{sec:problem form}
Consider $N>1$ autonomous agents, with $\mathcal{N} \coloneqq \{1,\dots,N\}$, operating in $\mathbb{R}^n$ and described by the spheres $\mathcal{A}_i(x_i) \coloneqq \{y\in\mathbb{R}^n : \|x_i-y\| < r_i \}$, with $x_i \in\mathbb{R}^n$ being agent $i$'s center, and $r_i\in\mathbb{R}_{> 0}$ its bounding radius. 
We consider that there exist $K>1$ points of interest in the workspace, denoted by $c_k\in\mathbb{R}^n$,  $\forall k\in\mathcal{K}\coloneqq \{1,\dots,K\}$, with $\Pi \coloneqq \{c_1,\dots,c_K\}$.
Moreover, we introduce disjoint sets of atomic propositions $\Psi_i$, expressed as boolean variables, that represent services provided by agent $i\in\mathcal{N}$ in $\Pi$. The services provided at each point $c_k$ are given by the labeling functions $\mathcal{L}_i:\Pi\rightarrow2^{\Psi_i}$, which assign to each point $c_k$, $k\in\mathcal{K}$, the subset of services $\Psi_i$ that agent $i$ can provide in that region. Note that, upon the visit to $c_k$, agent $i$ chooses among $\mathcal{L}_i(c_k)$ the subset of atomic propositions to be evaluated as true, i.e., the subset of services it \textit{provides} among the available ones. These services are abstractions of action primitives that can be executed in different regions, such as manipulation tasks or data gathering. In this work, we do not focus on how the service providing is executed by an agent; we only aim at controlling the agents' motion to reach the regions where these services are available.

The agents' motion is described by the following dynamics, inspired by rigid body motion:
\begin{subequations} \label{eq:dynamics}
	\begin{align}	
	&\dot{x}_i = v_i, \\
	&B_i \dot{v}_i + f_i(x_i,v_i) + g_i  = u_i,
	\end{align}
\end{subequations}
where  $v_i\in\mathbb{R}^n$ are the agents' generalized velocities,   $B_i\in\mathbb{R}^{n\times n}$ are positive definite matrices representing inertia, $g_i\in\mathbb{R}^n$ are gravity vectors, $u_i\in\mathbb{R}^n$ are the control inputs, and $f_i:\mathbb{R}^{2n}\to\mathbb{R}^n$ are terms representing modeling uncertainties, satisfying the following assumption.
%The terms $\dot{M}_i(x_i)-2C_i(x_i,v_i)$ are skew-symmetric and we further make the following assumption regarding $f_i(\cdot)$:
\begin{assumption} \label{ass:f_i}
	It holds that $\|f_i(x_i,v_i)\| \leq a_i\bar{f}_i(x_i)\|v_i\|$, $\forall (x_i,v_i)\in\mathbb{R}^{2n}$, $i\in\mathcal{N}$, where $a_i$ are \textit{unknown} positive constants and $\bar{f}_i:\mathbb{R}^{2n}\to\mathbb{R}_{\geq 0}$ are known continuous functions.
\end{assumption}

Moreover, we consider that each agent has a certain priority $\mathsf{pr}_i\in\mathbb{N}$ in the multi-agent team, with higher $\mathsf{pr}_i$ denoting higher priority. Without loss of generality, we assume that these variables have been normalized so that $\exists i\in\mathcal{N} : \mathsf{pr}_i = 1$ and
$|\mathsf{pr}_\ell - \mathsf{pr}_j| = 1$, $\forall \ell,j\in\mathcal{N}$, with $\ell \neq j$. %For instance in a $4$-agent team, the assignment $\mathsf{pr}_1 = 3$, $\mathsf{pr}_2 = 1$, $\mathsf{pr}_3 = 4$, $\mathsf{pr}_4 = 2$ means that agent $2$ has the highest priority, followed by agent $4$, etc. 
The priority variables can be given off-line to the agents. 

In addition, we consider that each agent has a limited sensing radius $d_{\text{con},i}\in\mathbb{R}_{>0}$, with $d_{\text{con},i} > \max_{j\in\mathcal{N}}\{r_i+r_j\}$, which implies that the agents can sense each other without colliding.
%\begin{assumption} [sensing radii] \label{ass:sensing radii}
%	Every agent $i\in\mathcal{N}$ has a limited sensing radius $d_{\text{con},i}\in\mathbb{R}_{>0}$, and it holds that $d_{\text{con},i} \geq  \max_{j\in\mathcal{N}\backslash\{i\}} \{ \sup_{(x_i,x_j)\in \mathbb{M}^2} \{ \|p_i - p_j\| :  \partial \mathcal{A}_i(x_i)\cap\partial \mathcal{A}_j(x_j) \neq \emptyset\} \} + \varepsilon_{\text{con}}$,
%	for an arbitrarily small positive constant $\varepsilon_{\text{con}}\in\mathbb{R}_{>0}$.
%\end{assumption}
Based on this, we model the topology of the multi-agent network through the undirected graph $\mathcal{G}(x) \coloneqq (\mathcal{N},\mathcal{E}(x))$, with  $\mathcal{E}(x) \coloneqq \{(i,j)\in\mathcal{N}^2 : \|x_i - x_j \| \leq \min\{d_{\text{con},i}, d_{\text{con},j}\} \}$. We further denote $M(x)\coloneqq |\mathcal{E}(x)|$. Given the $m$ edge in the edge set $\mathcal{E}(x)$, we use the notation $(m_1,m_2)\in\mathcal{N}^2$ that gives the agent indices that form edge $m\in\mathcal{M}(x)$, where $\mathcal{M}(x)\coloneqq\{1,\dots,M(x)\}$ is an arbitrary numbering of the edges $\mathcal{E}(x)$. By also denoting $m_1$ as the tail and $m_2$ as the head of edge $m$, we define the $N\times M$ incidence matrix $D(\mathcal{G}(x)) \coloneqq [d_{im}]$, where $d_{im} = 1$ if $i$ is the head of edge $m$, $d_{im} = -1$ if $i$ is the tail of edge $m$, and $d_{im} = 0$, otherwise. Note that, for a connected graph $\mathcal{G}$, the sum of the rows of $D(\mathcal{G})$ equals zero.
Next, we assume that  the agents form initially a collision-free connected graph.
\begin{assumption} \label{ass:initially connected}
	The graph $\mathcal{G}(x(0))$ is nonempty, connected and $\mathcal{A}_i(x_i(0))\cap\mathcal{A}_j(x_j(0)) = \emptyset$, $\forall i,j\in\mathcal{N}$, with $i\neq j$.
\end{assumption}
As mentioned before, the agents, apart from satisfying their local LTL formulas, need to (a) preserve connectivity with their initial neighbors, and (b) guarantee inter-agent collision avoidance. 
%Hence, given  Assumption \ref{ass:initially connected}, the proposed methodology will guarantee that $\mathcal{G}(x(t))$ will remain connected as well as that no inter-agent collisions occur, $\forall t\in\mathbb{R}_{>0}$. 
More specifically, we will guarantee that the initial edge set $\mathcal{E}(x(0))$ will be preserved and that $\mathcal{A}_i(x_i(t))\cap\mathcal{A}_j(x_j(t)) = \emptyset$, $\forall i,j\in\mathcal{N}$, with $i\neq j$, $t\in\mathbb{R}_{>0}$.

%Another important property concerning multi-agent systems is inter-agent collision avoidance, which is commonly not taken into account in hybrid frameworks involving temporal logic specifications. We aim, therefore, along with the satisfaction of the individual formulas and the connectivity constraints, to guarantee inter-agent collision.
In order to proceed, we need the following definitions: 
\begin{definition} \label{def:agent in region}
	An agent $i\in\mathcal{N}$, at configuration $x_i\in\mathbb{R}^n$, can provide a service at a point $c_k\in\mathbb{R}^n$, among the set $\mathcal{L}_i(\pi_k)$, if $c_k \in \mathcal{A}_i(x_i)$. 									 
\end{definition}
\begin{definition}
	Let $x_i(t)\in\mathbb{R}^n$, $t\in\mathbb{R}_{\geq 0}$, be a trajectory of agent $i\in\mathcal{N}$. The \textit{behavior} of agent $i$ is the tuple $\beta_i \coloneqq (c_{i1},\sigma_{i1}),(c_{i2},\sigma_{i2}),\dots$, with $c_{i\ell}\in \Pi$, $\forall \ell\in\mathbb{N},i\in\mathcal{N}$, and $c_{i\ell}\in\mathcal{A}_i(x_i(t))$, $\forall t\in \Delta t_{i\ell}\coloneqq [t_{i\ell},t'_{i\ell}] \subset \mathbb{R}_{\geq 0}$, $t_{i\ell} < t'_{i\ell} < t_{i(\ell+1)}$, $c_{k}\notin \mathcal{A}_i(x_i(t))$, $\forall k\in\mathcal{K}, t\in(t'_{i\ell},t_{i(\ell+1)})$, $\sigma_{i\ell}\in 2^{\Psi_i}$, $\sigma_{i\ell}\in(\mathcal{L}_i(c_{i\ell})\cup\emptyset)$.
\end{definition}
Loosely speaking, a behavior consists of the sequence of points $c_{i1}c_{i2}\dots$ where agent $i$ can provide services at, at the time intervals $\Delta t_{i\ell},\ell\in\mathbb{N}$. In every point $c_{i\ell}$, agent $i$ chooses to provide the set $\sigma_{i\ell}$ of services among the $\mathcal{L}_i(c_{i\ell})$ available ones. Note that $\sigma_{i\ell}$ can be the empty set, implying that the agent may choose not to provide any services. Given the agent's behavior $\beta_i$, the satisfaction of a task formula $\phi_i$ is defined as follows:
\begin{definition}
	A behavior $\beta_i$ satisfies $\phi_i$ if there exists a subsequence $\widetilde{\sigma}_i$ $\coloneqq$ $\sigma_{k_{i1}}$ $\sigma_{k_{i2}}$ $\dots$ of $\sigma_{i1}$ $\sigma_{i2}$ $\dots$, with $k_{i1}$, $k_{i2}$, $\dots$ being a subsequence of $i1, i2,\dots$, such that $\widetilde{\sigma}_i \models \phi_i$.
\end{definition}

The problem treated in this paper is the following:

\begin{problem}
	Consider $N$ spherical autonomous agents with dynamics \eqref{eq:dynamics} and $K$ points of interest in the workspace. Given the sets $\Psi_i$ and $N$ LTL formulas $\phi_i$ over $\Psi_i$, as well as Assumptions \ref{ass:f_i}-\ref{ass:initially connected}, develop a decentralized control strategy that achieves behaviors $\beta_i$, that yield the satisfaction of $\phi_i$, $\forall i\in\mathcal{N}$, while guaranteeing inter-agent collision avoidance and connectivity maintenance, i.e., $\mathcal{A}_i(x_i(t))\cap\mathcal{A}_j(x_j(t)) \neq \emptyset$, $\forall i,j\in\mathcal{N}$, with $i\neq j$, and $\|p_{m_1}(t) - p_{m_2}(t) \| \leq \min\{d_{\text{con},m_1},d_{\text{con},m_2}\}$, $\forall t\in\mathbb{R}_{\geq 0}, m\in\mathcal{M}(x(0))$.
%	\begin{itemize}
%		\item $\mathcal{A}_i(x_i(t))\cap\mathcal{A}_j(x_j(t)) \neq \emptyset$, $\forall t\in\mathbb{R}_{\geq 0}$, $i,j\in\mathcal{N}$, with $i\neq j$,
%		\item $\|p_{m_1}(t) - p_{m_2}(t) \| \leq \min\{d_{\text{con},m_1},d_{\text{con},m_2}\}$, $\forall t\in\mathbb{R}_{\geq 0}, m\in\mathcal{M}(x(0))$.
%	\end{itemize}
\end{problem} 

\section{Main Results}\label{sec:main results}
In this section we present the proposed solution, which consists of three layers: (i) an off-line plan synthesis for the discrete plan of each agent, i.e., the path of the goal points and the sequence of services to be provided; (ii) a distributed continuous control scheme that guarantees the navigation of one of the agents to a goal point of interest from \textit{all} collision-free and connected (in the sense of $\mathcal{E}(x(0))$) initial configurations; (iii) a decentralized hybrid control layer that coordinates the discrete plan execution via continuous control law switching, to ensure the satisfaction of each agent's local task.

\subsection{Discrete Plan Synthesis} \label{subsec:discrete plan synthesis}

The discrete plan can be generated using standard techniques from automata-based formal synthesis. We first model the motion of each agent as a finite transition system $\mathcal{T}_i \coloneqq (\Pi',c_{i,0},\to_i,\Psi_i,\mathcal{L}_i)$, where $c_{i,0}$ represents the agent's initial position $x_i(0)$, $\Pi'\coloneqq\Pi\cup\{c_{i,0}\}$ is the set of points of interest defined in Section \ref{sec:problem form}, expanded to include $c_{i,0}$, $\to_i\coloneqq \Pi\times\Pi$ is a transition relation, and $\Psi_i$, $\mathcal{L}_i$ are the sets of atomic propositions and labeling function, respectively, as defined in Section \ref{sec:problem form}. Note that, by the definition of the transition relation, we consider that there can be transitions between any pair of points of interest. This is achieved in the continuous time  motion by the proposed control scheme of the subsequent section. Next, each agent $i\in\mathcal{N}$ translates the LTL formula $\phi_i$ into a B\"uchi automaton $\mathcal{A}_{\phi_i}$ and builds the product $\widetilde{T}_i \coloneqq \mathcal{T}_i\otimes \mathcal{A}_{\phi_i}$. The accepting runs of $\widetilde{T}_i$ (that satisfy $\phi_i$) are projected onto $\mathcal{T}_i$ and provide for each agent a sequence of points to be visited and services to be provided in the prefix-suffix form: $\mathsf{plan}_i \coloneqq \ (c_{i1^\text{G}}, \sigma_{i1^\text{G}}) \ \dots \ (c_{il_i^\text{G}}, \sigma_{il_i^\text{G}}) \ ((c_{i(l_i+1)^\text{G}}, \sigma_{i(l_i+1)^\text{G}})\dots$ $(c_{iL_i^\text{G}}, \sigma_{iL_i^\text{G}}))^\mathsf{\omega}$, where $l_i,L_i \in \mathbb{N}$, with $l_i < L_i$, and $c_{i\ell^\text{G}}\in \Pi$, $\sigma_{i\ell^\text{G}}\in 2^{\Psi_i}, (\mathcal{L}_i(c_{i\ell^\text{G}})\cup\emptyset)$,  $\forall \ell\in\{1,\dots,L_i\}$, $i\in\mathcal{N}$. More details regarding the followed technique are beyond the scope of this paper and can be found in \cite{baier2008principles}.  Note that, in our work, LTL formulas are interpreted over the
provided services along a trajectory, not the available ones. Hence, crossing of points of interest not included in $\mathsf{plan}_i$ (which might happen due to the collision and connectivity constraints, as explained in the next sections) does not influence the local LTL task satisfaction.

\subsection{Continuous Control Design}\label{subsec:continuous control}
In this section we propose a decentralized control protocol for the transition of the agents to the points of interest, while guaranteeing inter-agent collision-avoidance and connectivity maintenance. More specifically, given a collision-free and connected (i.e., connected graph $\mathcal{G}(x(t_0))$) configuration of the agents at a time instant $t_0\in\mathbb{R}_{\geq 0}$, the proposed control scheme guarantees that exactly one agent $j\in\mathcal{N}$ navigates to a desired point, while preserving connectivity of the initial edge set and avoiding inter-agent collisions. Loosely speaking, connectivity maintenance forces the whole multi-agent team to navigate towards the desired point of agent $j$, while avoiding collisions. This is motivated by potential cooperative tasks of the agents at the points of interest (e.g. object transportation). Then, the hybrid coordination of the next section guarantees that all the agents will eventually reach their desired goals by an appropriate switching protocol based on the priority functions $\mathsf{pr}_i$. 

Let the points $c_i\in\mathbb{R}^n$, $\forall i\in\mathcal{N}$, be some desired destinations of the agents. Consider the initial connected graph $\mathcal{G}_0=(\mathcal{N},\mathcal{E}_0)\coloneqq \mathcal{G}(x(t_0)) = (\mathcal{N},\mathcal{E}(x(t_0)))$, with $M_0 \coloneqq M(x(t_0))$ and edge numbering $\mathcal{M}_0 \coloneqq \mathcal{M}(x(t_0))$. Consider also
the complete graph $\bar{\mathcal{G}} \coloneqq (\mathcal{N},\mathcal{E})$, with $\bar{\mathcal{E}}\coloneqq \{ (i,j), \forall i,j\in\mathcal{N} \text{ with } i < j\}$, $\bar{M}\coloneqq |\bar{\mathcal{E}}|$, and the edge numbering $\bar{\mathcal{M}}\coloneqq  \{1,\dots,M_0,M_0+1,\dots,\bar{M}\}$, where $\{M_0+1,\dots,\bar{M}\}$ corresponds to the edges in  $\bar{\mathcal{E}}\backslash \mathcal{E}_0$. In other words, we assume that the numbering of the extra edges $\bar{\mathcal{E}}\backslash\mathcal{E}_0$ starts from $M_0+1$.

Next, we construct the collision functions for all the edges $m\in\bar{\mathcal{M}}$. Let  $\beta_{\text{col},m}:\mathbb{R}_{\geq 0}\to[0,\bar{\beta}_\text{col}]$, with
\begin{align}
\beta_{\text{col},m}(x) \coloneqq \left\{ \begin{matrix}
\vartheta_{\text{col},m}(x) & 0 \leq x < \bar{d}_{\text{col},m}, \\
\bar{\beta}_\text{col} &	 \bar{d}_{\text{col},m} \leq x
\end{matrix} \right.,
\end{align}
where $\vartheta_{\text{col},m}:\mathbb{R}_{\geq 0}\to[0,\bar{\beta}_\text{col}]$ is a continuously differentiable \textit{strictly increasing} polynomial that renders $\beta_{\text{col},m}$ continuously differentiable, with $\vartheta_{\text{col},m}(0) = 0$, $\vartheta_{\text{col},m}(\bar{d}_{\text{col},m}) = \bar{\beta}_\text{col}$, $\forall m\in\bar{\mathcal{M}}$, and $\bar{\beta}_\text{col} $, $\bar{d}_{\text{col},m}$ are positive constants to be appropriately chosen. 
%satisfying $\bar{\Delta}_{\text{col},m} \leq \widetilde{\Delta}_{m_1,m_2}$ where $\widetilde{\Delta}_{m_1,m_2}   \coloneqq \inf_{(x_{m_1},x_{m_2})\in \mathbb{M}^2}\{ \Delta_{m_1,m_2}(x_{m_1},x_{m_2}) : \|p_{m_1} - p_{m_2} \| =   \min\{d_{\text{con},m_1},d_{\text{con},m_2} \} \}$,
 Then, for each edge $m\in\bar{\mathcal{M}}$, we can choose $\beta_{\text{col},m} \coloneqq$ $\beta_{\text{col},m}(\iota_m)$, where $\iota_m\coloneqq\|p_{m_1}-p_{m_2}\|^2 - (r_{m_1}+r_{m_2})^2$ and $\bar{d}_{\text{col},m}\coloneqq \underline{d}^2_{\text{con},m}- (r_{m_1}+r_{m_2})^2$, $\underline{d}_{\text{con},m} \coloneqq \min\{d_{\text{con},m_1},d_{\text{con},m_2}\}$, that vanishes when a collision between agents $m_1,m_2$ occurs.  The term $\bar{\beta}_\text{col}$ can be any positive constant.

Next, we construct the connectivity functions for all the edges $m\in\mathcal{M}$. Let $\beta_{\text{con},m}:\mathbb{R}_{\geq 0}\to[0,\bar{\beta}_\text{con}]$, with 
\begin{align*}
\beta_{\text{con},m}(x) \coloneqq \left\{ \begin{matrix}
\vartheta_{\text{con},m}(x) & 0 \leq x < \underline{d}^2_{\text{con},m} \\
\bar{\beta}_\text{con} &	 \underline{d}^2_{\text{con},m} \leq x \\
\end{matrix} \right.,
\end{align*}
where $\vartheta_{\text{con},m}:$ $\mathbb{R}_{\geq 0}$ $\to$ $[0,\bar{\beta}_\text{con}]$ is a cont. differentiable \textit{strictly increasing} polynomial that renders $\beta_{\text{con},m}$ continuously differentiable, with $\vartheta_{\text{con},m}(0) = 0$, $\vartheta_{\text{con},m}(\underline{d}^2_{\text{con},m}) = \bar{\beta}_\text{con}$, $\forall m\in\mathcal{M}$. Then, for each edge $m\in\mathcal{M}$, we choose $\beta_{\text{con},m} \coloneqq \beta_{\text{con},m}(\eta_m)$, with $\eta_m\coloneqq \underline{d}^2_{\text{con},m} - \|p_{m_1}-p_{m_2}\|^2$, that vanishes at a connectivity break of edge $m$. The term $\bar{\beta}_\text{con}$ can be any positive constant. 
The aforementioned functions take into account the limited sensing capabilities of the agents, since the derivatives of $\beta_{\text{col},m}$ and $\beta_{\text{con},m}$ are zero when $\|p_{m_1}-p_{m_2}\| \geq \underline{d}_{\text{con},m}$, $\forall m\in\bar{\mathcal{M}}$.
Note that all the necessary parameters for the construction of $\beta_{\text{col},m}$, $\beta_{\text{con},m}$ can be transmitted off-line to the agents $m_1,m_2$. 
Similarly to \cite{guo2016communication}, we propose now the following decentralized control scheme, parameterized by the goal and mode of the agents:
\begin{align}
	&u_i(c_i,\mathsf{md}_i) \coloneqq \sum\limits_{m\in\bar{\mathcal{M}}}\alpha_{\text{col}}(i,m) \beta'_{\text{col},m}\frac{\partial \iota_m}{\partial x_{m_1}} + \notag \\
	& +\sum\limits_{m\in\mathcal{M}_0}\alpha_\text{con}(i,m)  \beta'_{\text{con},m}\frac{\partial \eta_m}{\partial x_{m_1}} - \mathsf{md}_i  \gamma_i(c_i)+ g_i    \notag \\
	& -\Big(\hat{a}_i\bar{f}_i(x_i) + \mu_i\Big)v_i,	\label{eq:control law}
\end{align}
where $c_i\in\mathbb{R}^n$ is agent $i$'s desired destination, $\mathsf{md}_i\in\{0,1\}$ is the agent's mode (active or passive); the functions $\alpha_{\text{col}}$, $\alpha_{\text{con}}$ are defined as $\alpha_\text{col}(i,m) = -\mu_{\text{col},m}$ if $i=m_1$ (agent $i$ is the tail of edge $m$), $\alpha_\text{col}(i,m) = \mu_{\text{col},m}$ if $i=m_2$ (agent $i$ is the head of edge $m$), and $\alpha_\text{col}(i,m) = 0$ otherwise, $\alpha_\text{con}(i,m) = -\mu_{\text{con},m}$ if $i=m_1$, $\alpha_\text{con}(i,m) = \mu_{\text{con},m}$ if $i=m_2$, and $\alpha_\text{con}(i,m) = 0$ otherwise, $i\in\mathcal{N}$; $\beta'_{\text{col},m} \coloneqq \frac{\partial }{\partial \iota_m}\left(\frac{1}{\beta_{\text{col},m}(\iota_m)}\right)$, $\beta'_{\text{con},m} \coloneqq \frac{\partial }{\partial \eta_m}\left(\frac{1}{\beta_{\text{con},m}(\eta_m)}\right)$, $\gamma_i(c_i)\coloneqq \mu_{c,i}(x_i - c_i)$; the constants $\mu_{\text{col},m}, \mu_{\text{con},m},\mu_{c,i}, \mu_i\in\mathbb{R}_{>0}$ are positive gains, $\forall m\in\bar{\mathcal{M}}$, $m\in\mathcal{M}_0$, $i\in\mathcal{N}$,
and the terms $\hat{a}_i$ are adaptation signals that evolve according to 
\begin{equation}
	\dot{\hat{a}}_i = \mu_{a,i}\bar{f}_i(x_i)\|v_i\|^2, \label{eq:adaptation laws}
\end{equation}
with arbitrary bounded initial conditions $\hat{a}_i(t_0)$, and positive gains $\mu_{a,i}\in\mathbb{R}_{>0}$, $\forall i\in\mathcal{N}$. 
%Note that $\frac{\partial \eta_m}{\partial m_1} = - \frac{\partial \eta_m}{\partial m_2}$ and $\frac{\partial \iota_m}{\partial m_1} = - \frac{\partial \iota_m}{\partial m_2}$. T
The intuition behind the parameters $\mathsf{md}_i$ is that only one of them can be true at time, meaning that only one agent navigates towards its desired point. After a successful navigation,  the variable is activated for another agent, and so on. Section \ref{subsec:hybrid strategy} describes the coordination strategy that decides about the activation of the variables $\mathsf{md}_i$. The navigation of the agent $j$ for which $\mathsf{md}_j = 1$ is guaranteed by the next theorem. 
	
\begin{theorem}
	Consider a multi-agent team $\mathcal{N}$, described by the dynamics \eqref{eq:dynamics}, at a collision-free and connected configuration at $t=t_0\in\mathbb{R}_{\geq 0}$, with desired destinations $c_i$, $\forall i\in\mathcal{N}$. Then, under Assumptions \ref{ass:f_i}-\ref{ass:initially connected},  the application of the control laws \eqref{eq:control law} with $u_j = u_j(c_j,1)$ for a $j\in\mathcal{N}$ and $u_i=u_i(c_i,0)$, $\forall i\in\mathcal{N}\backslash\{j\}$ guarantees that $c_j\in\mathcal{A}_j(x_j(t_f))$ for a finite $t_f$, as well as $\mathcal{A}_i(x_i(t))\cap\mathcal{A}_n(x_n(t))=\emptyset$, $\forall i,n\in\mathcal{N}$, with $i\neq n$, and $\|p_{m_1}(t) - p_{m_2}(t) \| \leq \min\{d_{\text{con},m_1},d_{\text{con},m_2}\}$, $\forall t\geq t_0, m\in\mathcal{M}_0$, with bounded closed loop signals.
\end{theorem}
\begin{proof}
		By taking into account that $\frac{\partial \iota_m}{\partial x_{m_1}} =- \frac{\partial \iota_m}{\partial x_{m_2}}$, $\forall m\in\bar{\mathcal{M}}$, $\frac{\partial \eta_m}{\partial x_{m_1}} = -\frac{\partial \eta_m}{\partial x_{m_2}}$, $\forall m\in\mathcal{M}_0$, we can write the control laws \eqref{eq:control law} in vector form:
		\begin{align}
		u =& (D(\mathcal{G}_0)\otimes I_n)\mu_{\text{con}}\beta_{\text{con}} + (D(\bar{\mathcal{G}})\otimes I_n)\mu_{\text{col}}\beta_{\text{col}} -\gamma_{\mathsf{md}}(x)\notag \\
		&  + g - h(x)v  \label{eq:control law vector form}
		\end{align}
		where $g \coloneqq [g_1^\top,\dots,g_N^\top]^\top$, $x \coloneqq [x_1^\top,\dots,x_N^\top]^\top$, $v \coloneqq [v_1^\top,\dots,v_N^\top]^\top \in\mathbb{R}^{Nn}$, $h(x) = \text{diag}\{[\hat{a}_i\bar{f}_i(x_i)+\mu_i]_{i\in\mathcal{N}}\}\in\mathbb{R}^{Nn\times Nn}$, $\gamma_{\mathsf{md}}(x)\in\mathbb{R}^{Nn}$ is a vector of zeros except for the rows $nj,\dots, n(j+1)$, which are $\gamma_j(c_j)$; 
		$\mu_{\text{con}} \coloneqq \text{diag}\{ [\mu_{\text{con},m} I_n]_{m\in\mathcal{M}_0} \}$, $\mu_{\text{col}} \coloneqq \text{diag}\{ [\mu_{\text{col},m} I_n]_{m\in\bar{\mathcal{M}}}\}\in\mathbb{R}^{Nn\times Nn}$,   $D(\cdot)$ is the graph incidence matrix, as defined in Section \ref{sec:problem form}, and $\beta_\text{con} \coloneqq \left[\beta'_{\text{con},1}\left(\frac{\partial \eta_1}{\partial x_{1_1}}\right)^\top,\dots,\beta'_{\text{con},M_0}\left(\frac{\partial \eta_{M_0}}{\partial x_{(M_0)_1}}\right)^\top \right]^\top \in \mathbb{R}^{nM_0}$, $\beta_\text{col} \coloneqq \left[\beta'_{\text{col},1}\left(\frac{\partial \iota_1}{\partial x_{1_1}}\right)^\top,\dots,\beta'_{\text{col},\bar{M}}\left(\frac{\partial \iota_{\bar{M}}}{\partial x_{\bar{M}_1}}\right)^\top \right]^\top \in \mathbb{R}^{n\bar{M}}$.

	Consider the positive definite Lyapunov candidate $V(x,v,\hat{a}) \coloneqq\frac{\mu_{c,j}}{2}\|x_j-c_j\|^2 + \frac{1}{2}\sum_{i\in\mathcal{N}}\Big(v_i^\top B_i v_i + \frac{1}{2\mu_{a,i}}\widetilde{a}_i^2  \Big) 
	+\sum_{m\in\bar{\mathcal{M}}}  \frac{\mu_{\text{col},m}}{\beta_{\text{col},m}(\iota_m)}  +
	\sum_{m\in\mathcal{M}_0}\frac{\mu_{\text{con},m}}{\beta_{\text{con},m}(\eta_m)}$, 
	%\begin{align*}
	%	&\red{V(x,v,\hat{a})} \coloneqq\frac{\mu_{c,j}}{2}\|x_j-c_j\|^2 + \frac{1}{2}\sum\limits_{i\in\mathcal{N}}\Big(v_i^\top B_i v_i + \frac{1}{2\mu_{a,i}}\widetilde{a}_i^2  \Big) \notag \\
	%	&+\sum\limits_{m\in\bar{\mathcal{M}}}  \frac{\mu_{\text{col},m}}{\beta_{\text{col},m}(\iota_m)}  +
	%	\sum\limits_{m\in\mathcal{M}_0}\frac{\mu_{\text{con},m}}{\beta_{\text{con},m}(\eta_m)},
	%\end{align*}
	where $\hat{a} = [\hat{a}_1,\dots,\hat{a}_N]^\top\in\mathbb{R}^N$, and $\widetilde{a}_i\coloneqq\hat{a}_i-a_i$, $\forall i\in\mathcal{N}$.
	The connectedness of $\mathcal{M}_0$ and collision-free initial conditions imply the existence of a finite constant $\bar{V}$ such that $V(t_0) \leq \bar{V}$. By taking the derivative of $V$ we obtain $\dot{V} = \gamma_j(c_j)^\top v_j + \sum_{i\in\mathcal{N}}\Big\{\widetilde{a}_i\bar{f}_i(x_i)\|v_i\|^2 + v_i^\top(u_i - g_i - 
	 f_i(x_i,v_i)) \} 	- \Big(\widetilde{\beta}_\text{con}^\top(D(\mathcal{G}_0)\otimes I_n)^\top + \widetilde{\beta}_\text{col}^\top(D(\bar{\mathcal{G}})\otimes I_n)^\top\Big) v$,
	%\small
	%\begin{align*} 
	%	&\dot{V} = \gamma_j(c_j)^\top v_j + \sum_{i\in\mathcal{N}}\Big\{\widetilde{a}_i\bar{f}_i(x_i)\|v_i\|^2 + v_i^\top(u_i - g_i - \\
	%	& f_i(x_i,v_i)) \} 	- \Big(\widetilde{\beta}_\text{con}^\top(D(\mathcal{G}_0)\otimes I_n)^\top + \widetilde{\beta}_\text{col}^\top(D(\bar{\mathcal{G}})\otimes I_n)^\top\Big) v, 
	%\end{align*}
	%\normalsize
	where $\widetilde{\beta}_\text{con}\coloneqq \mu_\text{con}\beta_\text{con}$, $\widetilde{\beta}_\text{col}\coloneqq \mu_\text{col}\beta_\text{col}$. By substituting the control and adaptation laws \eqref{eq:control law vector form}, \eqref{eq:adaptation laws} and employing Assumption \ref{ass:f_i}, we obtain $\dot{V} \leq\sum_{i\in\mathcal{N}}\{\|v_i \|\|f_i(x_i,v_i)\| -\hat{a}_i \bar{f}_i(x_i)\|v_i\|^2 + \widetilde{a}_i\bar{f}_i(x_i)\|v_i\|^2 - \mu_i\|v_i\|^2 \} \leq  \sum_{i\in\mathcal{N}} \{a_i\bar{f}_i(x_i)\|v_i\|^2 -\hat{a}_i \bar{f}_i(x_i)\|v_i\|^2 + \widetilde{a}_i\bar{f}_i\|v_i\|^2 - \mu_i\|v_i\|^2\} =  -\sum_{i\in\mathcal{N}}\mu_i\|v_i\|^2$. Hence, we conclude that $\dot{V} \leq 0$, which implies that $V(t) \leq V(t_0) \leq \bar{V}$. Therefore, we conclude that  $\beta_{\text{col},m}(\iota_m) \geq \frac{\mu_{\text{col},m}}{\bar{V}}$ $\beta_{\text{con},m}(\eta_m) \geq \frac{\mu_{\text{con},m}}{\bar{V}}$, and 
	$\|x_j - c_j\| \leq \frac{2\bar{V}}{\mu_{c,j}}$, i.e., the boundedness of $x_j$ (since $c_j$ is finite), the boundendess of $v_i, \hat{a}_i$, $\forall i\in\mathcal{N}$, as well as that the multi-agent trajectory is free of collisions and connectivity breaks, $\forall t\geq t_0$. Since the multi-agent system stays connected and $x_j$ is bounded, we conclude that the rest $x_i$, $i\in\mathcal{N}\backslash\{j\}$ are also bounded, $\forall t\geq t_0$. Moreover, by invoking LaSalle's invariance principle, we conclude that the system will converge to the largest invariant set contained in $\mathbb{L} \coloneqq \{ (x,v,\hat{a})\in\mathbb{R}^{2Nn} : v_i = 0, \forall i\in\mathcal{N}\}$, which is the set $\widetilde{\mathbb{L}} \coloneqq \{ ((x,v,\hat{a}))\in\mathbb{R}^{2Nn} : \dot{v}_i = 0,v_i = 0, \forall i\in\mathcal{N}\}$. By considering the closed loop system \eqref{eq:dynamics}-\eqref{eq:control law vector form} and taking into account the positive definiteness of $B_i$, we conclude that the system will converge to the configuration 
	%$\widetilde{\mathbb{L}} = \{(x,v)\in\mathbb{R}^{6N} : (D(\mathcal{G}_0)\otimes I_3)\mu_{\text{con}}\beta_{\text{con}} + (D(\bar{\mathcal{G}})\otimes I_3)\mu_{\text{col}}\beta_{\text{col}} -\gamma_{\mathsf{md}}(x)  =0  \}$. 
	%It can be proven that the nullspace of $\widetilde{E}_\zeta(\cdot)$ is empty, and since $\zeta_i$ are unit quaternions, we conclude that the system will converge to the configuration where
%	\begin{align*}
%		\dot{V} \leq& \sum\limits_{i\in\mathcal{N}}\Big\{\|v_i \|\|f_i(x_i,v_i)\| -\hat{a}_i \bar{f}_i(x_i)\|v_i\|^2 + \widetilde{a}_i\bar{f}_i(x_i)\|v_i\|^2 - \mu_i\|v_i\|^2 \Big\} \\
%		& \sum\limits_{i\in\mathcal{N}} a_i\bar{f}_i(x_i)\|v_i\|^2 -\hat{a}_i \bar{f}_i(x_i)\|v_i\|^2 + \widetilde{a}_i\bar{f}_i\|v_i\|^2
%	\end{align*}	
	\begin{equation}
		(D(\mathcal{G}_0)\otimes I_n)\widetilde{\beta}_\text{con}+ (D(\bar{\mathcal{G}})\otimes I_n)\widetilde{\beta}_\text{col} - \gamma_{\mathsf{md}}(x) =0. \label{eq:LaSalle}
	\end{equation}
	Note that $\mathcal{G}_0$ and $\bar{\mathcal{G}}$ are connected graphs, and hence the sum of the rows of $D(\mathcal{G}_0)$ and $D(\bar{\mathcal{G}})$ is zero. In particular, let $D(\mathcal{G}_0) = [d_{0,1},\dots,d_{0,N}]^\top$, $D(\bar{\mathcal{G}})= [\bar{d}_1,\dots, \bar{d}_{N}]^\top$, where $d^\top_{0,i}\in\mathbb{R}^{M_0}$, $\bar{d}^\top_i\in\mathbb{R}^{\bar{M}}$, $i\in\mathcal{N}$,  are the rows of $D(\mathcal{G}_0)$ and $D(\bar{\mathcal{G}})$, respectively. Then it holds that $\sum_{i\in\mathcal{N}}d_{0,i} = \sum_{i\in\mathcal{N}}\bar{d}_i = 0$. We can then write $D(\mathcal{G}_0)\otimes I_n = [d_{0,1}\otimes I_n,\dots,d_{0,N}\otimes I_n]^\top $, $D(\bar{\mathcal{G}})\otimes I_n= [\bar{d}_1\otimes I_n,\dots, \bar{d}_{N}\otimes I_n]^\top$   	
	 %(and thus of $D(\mathcal{G}_0)\otimes I_7$ and $D(\bar{\mathcal{G}})\otimes I_7$) equals to zero. Let $D(\mathcal{G}_0)\otimes I_7 = [d_{0,1},\dots, d_{0,7N}]^\top$, $D(\bar{\mathcal{G}})\otimes I_7 = [\bar{d}_1,\dots, \bar{d}_{7N}]^\top$, where $d^\top_{0,\ell}\in\mathbb{R}^{7M_0}$, $\bar{d}^\top_\ell\in\mathbb{R}^{7\bar{M}}$, $\ell\in\widetilde{\mathcal{N}}\coloneqq\{1,\dots,7N\}$,  are the rows of $D(\mathcal{G}_0)\otimes I_7$ and $(\bar{\mathcal{G}})\otimes I_7$, respectively. Then it holds that $\sum_{\ell\in\widetilde{\mathcal{N}}}d_{0,\ell} = \sum_{\ell\in\widetilde{\mathcal{N}}}\bar{d}_\ell = 0$, and, therefore, $\sum_{\ell\in\widetilde{\mathcal{N}}}\widetilde{d}_\ell = 0$, where $\widetilde{d}_\ell \coloneqq [d^\top_{0,\ell}, \bar{d}_\ell]^\top\in\mathbb{R}^{14M_0+7\bar{M}}$. 	 
	and hence \eqref{eq:LaSalle} becomes
	\begin{subequations}\label{eq:LaSalle element form}
	\begin{align}
		& [d_{0,i}\otimes I_n]^\top \widetilde{\beta}_\text{con} + [\bar{d}_i\otimes I_n]^\top \widetilde{\beta}_\text{col} = 0, \  \forall i \in \mathcal{N}\backslash\{j\}, \label{eq:LaSalle element form 1} \\
		& \gamma_j(c_j) - [d_{0,j}\otimes I_n]^\top \widetilde{\beta}_\text{con} - [\bar{d}_j\otimes I_n]^\top \widetilde{\beta}_\text{col} = 0.  \label{eq:LaSalle element form 2}
	\end{align}
	\end{subequations}
	From \eqref{eq:LaSalle element form 2} we obtain that $\gamma_j(c_j) - [(-\sum_{i\in\mathcal{N}\backslash\{j\}} d_{0,i} )\otimes I_n]^\top\widetilde{\beta}_\text{con} - [(-\sum_{i\in\mathcal{N}\backslash\{j\}} \bar{d}_i )\otimes I_n]^\top\widetilde{\beta}_\text{col} = 0$, which implies $\gamma_j(c_j) + \sum_{i\in\mathcal{N}\backslash\{j\}}\{ [d_{0,i}\otimes I_n]^\top\widetilde{\beta}_\text{con} + [\bar{d}_i\otimes I_n]^\top\widetilde{\beta}_\text{col} \} = 0$ and in view of \eqref{eq:LaSalle element form 1},  $\gamma_j(c_j)= 0$. Therefore, it holds that $\lim_{t\to\infty}x_j(t) = c_j$, which implies that, for every $\varepsilon$, there exists a $t_f > t_0$ such that $\|x_j(t) - c_j\| < \varepsilon, \forall t \geq t_f$. Hence, since $x_j$ is the center of $\mathcal{A}_j(x_j)$, we conclude that there exists a finite $t_f$ such that $c_j\in\mathcal{A}_j(x_j(t_f))$, which leads to the conclusion of the proof.
%	\begin{align}
%		&\gamma_j(c_j) - [(-\sum\limits_{i\in\mathcal{N}\backslash\{j\}} d_{0,i} )\otimes I_7]^\top\widetilde{\beta}_\text{con} - \notag \\ &[(-\sum\limits_{i\in\mathcal{N}\backslash\{j\}} \bar{d}_i )\otimes I_7]^\top\widetilde{\beta}_\text{col} = 0  \Leftrightarrow \notag \\
%		&\gamma_j(c_j) + \sum\limits_{i\in\mathcal{N}\backslash\{j\}}\Big\{ [d_{0,i}\otimes I_7]^\top\widetilde{\beta}_\text{con} + [\bar{d}_i\otimes I_7]^\top\widetilde{\beta}_\text{col} \Big\} = 0
%	\end{align}	
\end{proof}

%\begin{remark}
%\red{ Note that, compared to the potential field defined in [CITE MENG], the proposed control laws guarantee the navigation of agent $j$ to point $c_j$ under \textit{all} collision-free and connected initial conditions.  }
%\end{remark}

\subsection{Hybrid Control Strategy} \label{subsec:hybrid strategy}
In this section, we propose a decentralized switching strategy for each agent to decide on its own activity or passivity. Through this strategy, we integrate the discrete plan execution from Section \ref{subsec:discrete plan synthesis} and the continuous control scheme from Section \ref{subsec:continuous control} into a hybrid control scheme, which monitors the plan execution online. The desired plans for the agents, from Section \ref{subsec:discrete plan synthesis}, are  $\mathsf{plan}_i \coloneqq \ (c_{i1^\text{G}}, \sigma_{i1^\text{G}}) \ \dots \ (c_{il_i^\text{G}}, \sigma_{il_i^\text{G}}) \ ((c_{i(l_i+1)^\text{G}}, \sigma_{i(l_i+1)^\text{G}})\dots$ $(c_{iL_i^\text{G}}, \sigma_{iL_i^\text{G}}))^\mathsf{\omega}$, i.e., agent $i\in\mathcal{N}$, has to pass through the points $c_{i1^\text{G}}, \dots, c_{iL_i^\text{G}}$ and provide the corresponding services $\sigma_{i1^\text{G}},\dots,\sigma_{iL_i^\text{G}}$, which satisfy formula $\phi_i$, i.e, $\sigma_{i1^\text{G}}\dots\sigma_{il_i^\text{G}}(\sigma_{i(l_1+1)^\text{G}} \sigma_{iL_i^\text{G}})^\mathsf{\omega}\models \phi_i$. 

Let each agent have a counter variable $s_i$ initiated at $s_i = 1$, as well as a cycle counter $\kappa_i$, initiated at $\kappa_i = 1$, $\forall i\in\mathcal{N}$. Then, given the agent priority variables $\mathsf{pr}_i$, each agent executes $u_i = u_i(c_{is_i^\text{G}},1)$ if $\kappa_i = \mathsf{pr}_i$ and $u_i(c_{is_i^\text{G}},0)$ if $\kappa_i \neq \mathsf{pr}_i$. The agents update the cycle counter $\kappa_i$ every time the current active agent reaches its desired point, and the variable $s_i$ every time they reach their current desired point. Each agent provides the services $\sigma_{il^\text{G}}$ if $c_{il^\text{G}}\in\mathcal{A}_i(x_i)$ and $\kappa_i = \mathsf{pr}_i$, otherwise he does not provide any services.  More specifically, we construct the following algorithm: 
%\begin{algorithm}[H]
%\caption{Hybrid Control Strategy}
% \begin{algorithmic}[1]
%%\\ \textit{Initialization}:\\
%\STATE $\kappa_i = 1, s_i = 1$, $\forall i\in\mathcal{N}$
%\FOR {$i\in\mathcal{N}$}
%\IF {$\kappa_i = \mathsf{pr}_i$}
%\STATE $\mathsf{cur} = i$, \hspace{3mm} $u_i = u_i(c_{is_i^\text{G}},1)$
%%\STATE $u_i = u_i(c_{is_i^\text{G}},1)$, 
%\ELSE 
%\STATE $u_i=u_i(c_{is_i^\text{G}},0)$, 
%\ENDIF
%\ENDFOR
%\FOR {$i\in\mathcal{N}$}
%\IF	{$c_{s_{\mathsf{cur}}}\in\mathcal{A}_\mathsf{cur}(x_{\mathsf{cur}})$}
%\STATE Agent $\mathsf{cur}$ provides services $\sigma_{is_i^\text{G}}$
%\STATE $\kappa_i = (\kappa_i + 1)\mod N$ 
%\IF	{$s_\mathsf{cur} < \red{L^\text{G}_\mathsf{cur}}$}
%\STATE $s_\mathsf{cur} = s_\mathsf{cur} + 1$
%\ELSE	   
%\STATE $s_\mathsf{cur} = (s_\mathsf{cur} + 1)\mod \red{L^\text{G}_\mathsf{cur} + l^\text{G}_\mathsf{cur}}$
%\ENDIF
%\ENDIF
%%\STATE $u_j = u_j(c_{j1^\text{G}},1), \text{ for } j:\mathsf{pr}_j = 1$, $u_i=u_j(c_{j1^\text{G}},0), $
%\ENDFOR
%\end{algorithmic}
%\end{algorithm}

\begin{algorithm}	
		%\\ \textit{Initialization}:\\
		$\kappa_i \leftarrow 1, s_i \leftarrow 1$, $\forall i\in\mathcal{N}$\\
		\For{$i\in\mathcal{N}$}{
			\If{$\kappa_i = \mathsf{pr}_i$}
				{$\mathsf{cur} \leftarrow i$, \hspace{3mm} $u_i \leftarrow u_i(c_{is_i^\text{G}},1)$}
			\Else{
				$u_i \leftarrow u_i(c_{is_i^\text{G}},0)$}}
		\For{$i\in\mathcal{N}$}{
			\If	{$c_{s_{\mathsf{cur}}}\in\mathcal{A}_\mathsf{cur}(x_{\mathsf{cur}})$}
				{Agent $\mathsf{cur}$ provides services $\sigma_{is_i^\text{G}}$\\
				$\kappa_i \leftarrow (\kappa_i + 1)\mod N$ \\
			\If	{$s_\mathsf{cur} < L^\text{G}_\mathsf{cur}$}{
				$s_\mathsf{cur} \leftarrow s_\mathsf{cur} + 1$}
			\Else	   
				{$s_\mathsf{cur} \leftarrow (s_\mathsf{cur} + 1)\mod L^\text{G}_\mathsf{cur} + l^\text{G}_\mathsf{cur}$}}}
		%\STATE $u_j = u_j(c_{j1^\text{G}},1), \text{ for } j:\mathsf{pr}_j = 1$, $u_i=u_j(c_{j1^\text{G}},0), $		
		\caption{Hybrid Control Strategy}
\end{algorithm}

Loosely speaking, agent $i$ provides the services $\sigma_{is^{\text{G}}}$ only if $c_{is_i^\text{G}}\in\mathcal{A}_i(x_i)$, i.e., if it is in the respective desired point of interest, and $\kappa_i = \mathsf{pr}_i$, i.e., it is its turn to be active. As soon as it reaches the point and provides the services, it updates its progressive goal index $s_i$, and everyone in the team updates the cycle counter $\kappa_i$, so that another agent becomes active. Note that the agents need to know when the current agent reaches its progressive goal and provides its services so that they update the counter variable $\kappa_i$. To that end, the current agent can simply communicate this information as soon as it provides its services. Since the communication graph is always connected, the information can propagate to all agents. Note that potential time delays in this inter-agent communication do not affect the overall strategy. 
%The first agent to switch its control protocol is the one that switches from active ($u_i(\cdot,1)$) to passive mode ($u_i(\cdot,1)$)
A communication-free solution could be the use of state and input estimators along with
the discontinuous change of the control law of the current agent \cite{guo2016communication}.

%That can achieved 
%by employing the continuous control analysis of Section \ref{subsec:continuous control}. In particular, note that when an agent approaches its goal, i.e. $p_j(t) \to c_{js_j^\text{G}}$ for some $j\in\mathcal{N}$, LaSalle's invariance principle suggests that all the agents' velocities converge asymptotically to zero, i.e., $v_i(t) \to 0$, $\forall i\in\mathcal{N}$, which implies that the agent converge to a local minima. Hence, a possible trigger strategy for the agents is to measure the time duration that their velocity is below a certain threshold $\varepsilon_i$, and update their counters $\kappa_i$ and move to the next cycle as soon as this duration exceeds a specific constant $T_i$. In other words, let $\|v_i(t_i^-)\| > \varepsilon_i$, and $\| v_i(t) \| \leq \varepsilon_i$, $\forall t\geq t_i$. Then $t = T_i+t_i \Rightarrow \kappa_i =(\kappa_i+1)\mod N$, i.e., update $\kappa_i$ the time instant $T_i+t_i$. Note that $T_i$ and $\varepsilon_i$ need to be sufficiently large and small, respectively, to correspond to the aforementioned local minima. Note also that there's no need for the agents to update their $\kappa_i$ and move to the next cycle simultaneously. 

In that way, all the agents eventually reach their goal points of interest and provide the corresponding services. More specifically, the resulting time trajectory of each agent yields the behavior $\beta_i = (c_{i1},\sigma_{i1})(c_{i2},\sigma_{i2})\dots$, and the desired behavior $\mathsf{plan}_i$ is a subsequence of $\beta_i$, with $\sigma_{i\ell} = \emptyset$, $\forall \ell : \sigma_{i\ell}\neq \sigma_{i\ell^\text{G}}$, i.e., agent $i$ does not provide any services in unplanned crossing of points of interest (while navigating to a desired point or being in passive mode), 
providing only the desired services at the corresponding desired points. 

\begin{figure}[t!]
	\centering
	\includegraphics[scale = 0.48, trim = 0cm 0cm 0cm -0.5cm]{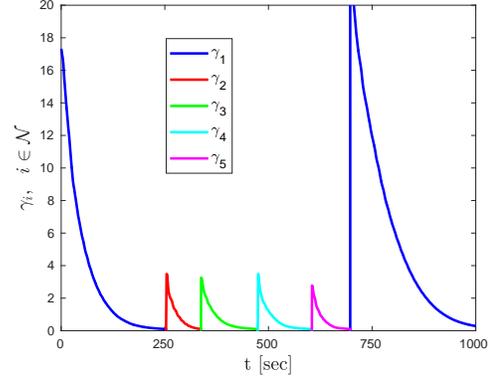}
	\caption{The distance errors $\mathsf{md}_i\gamma_i$, $\forall i\in\mathcal{N}$, $t\in[0,10^3]$.}
	\label{fig:gammas}
\end{figure}

\begin{figure}[t!]
	\centering
	\includegraphics[scale = 0.48, trim = 0cm 0cm 0cm 0.35cm]{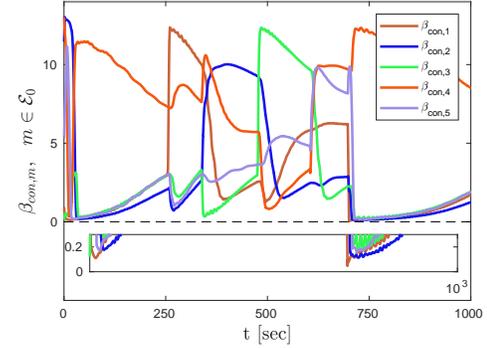}
	\caption{The functions $\beta_{\text{con},m}(\eta_m)$, $\forall m\in\mathcal{E}_0$, $t\in[0,10^3]$.}
	\label{fig:beta_con}
\end{figure}

\begin{figure}[t!]
	\centering
	\includegraphics[scale = 0.48,trim = 0cm 0cm 0cm -0.5cm]{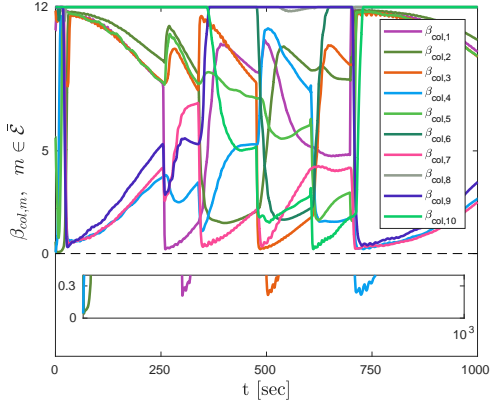}
	\caption{The functions $\beta_{\text{col},m}(\iota_m)$, $\forall m\in\bar{\mathcal{E}}$, $t\in[0,10^3]$.}
	\label{fig:beta_col}
\end{figure}

\begin{figure}[t!]
	\centering
	\includegraphics[scale = 0.48, trim = 0cm 0cm 0cm 0.35cm]{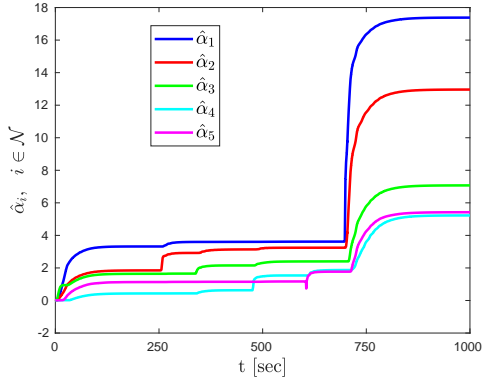}
	\caption{The adaptation signals $\hat{\alpha}_i$, $\forall i\in\mathcal{N}$, $t\in[0,10^3]$.}
	\label{fig:a_hats}
\end{figure}
%%%%%%%%%%%%%%%%%%%%%%%%%%%%%%%%%%%%%%%%%%%%%%%%%%%%%%%%%%%%%%%%%%%%%%%%%%%%%%%%
\section{SIMULATION RESULTS}\label{sec:Simulation}
We consider $N = 5$ holonomic spherical agents in $\mathbb{R}^3$, with $r_i=1\text{m}$, $d_{\text{con},i} = 4\text{m}$, priorities as $\mathsf{pr}_i = i$, $\forall i\in\mathcal{N}$, and initial positions $x_1 = [0,0,0]^\top\text{m}$, $x_2 = [-2.1,-2.3,2]^\top\text{m}$, $x_3 = [1.3,1.3,1.5]^\top\text{m}$, $x_4 = [-2,3.25,2.2]^\top\text{m}$, $x_5 = [2,2.4,-0.15]^\top\text{m}$, which give the edge set $\mathcal{E}_0 = \{(1,2),(1,3),(3,4),(3,5),(1,5)\}$. The complete edge set is $\bar{\mathcal{E}}$ $=$ $\{(1,2)$, $(1,3)$, $(3,4)$, $(3,5)$, $(1,5)$, $(1,4)$, $(2,3)$, $(2,4)$, $(2,5)$, $(4,5)\}$. We choose $B_i = b_{m_i} I_3$ and $f_i(x_i,v_i) = \alpha_i\|x_i\|\sin(w_{i,1}t + w_{i,2})v_i$, with $b_{m_i}, w_{i,1}$, $w_{i,2}$ randomly chosen in the interval $(1,2)$, $\forall i\in\mathcal{N}$. 
The points of interest are $c_1 = [10,10,10]^\top \text{m}$, $c_2 = [-5,0,5]^\top \text{m}$, $c_3 = [5,-2,-7]^\top \text{m}$, $c_4 = [0,-6,2]^\top \text{m}$. For simplicity, we consider that each agent can provide the services $\Psi_i = \{``\mathsf{r}_i",``\mathsf{b}_i",``\mathsf{g}_i",``\mathsf{m}_i"\}$, $\forall i\in\mathcal{N}$, and $\mathcal{L}_i(c_1) = \{``\mathsf{r}_i"\}$, $\mathcal{L}_i(c_2) = \{``\mathsf{b}_i"\}$, $\mathcal{L}_i(c_3) = \{``\mathsf{g}_i"\}$, $\mathcal{L}_i(c_4) = \{``\mathsf{m}_i"\}$ $\forall i\in\mathcal{N}$. The LTL formulas were taken as $\phi_1 = \square\lozenge(``\mathsf{r}_1"\land``\mathsf{r}_1"\bigcirc ``\mathsf{g}_1"\bigcirc \mathsf{m}_1 \bigcirc ``\mathsf{b}_1")$, $\phi_2 = \lozenge``\mathsf{m}_2"\land \square\lozenge(``\mathsf{r}_2"\land``\mathsf{b}_2")$, $\phi_3 = \lozenge``\mathsf{m}_3"\land \square\lozenge(``\mathsf{r}_3"\land``\mathsf{b}_3")$, $\phi_4 = \square\lozenge(``\mathsf{g}_4"\land``\mathsf{g}_4"\bigcirc ``\mathsf{b}_4"\bigcirc \mathsf{m}_4 \bigcirc ``\mathsf{g}_4")$, and $\phi_5 = ``\mathsf{r}_5"\land \square\lozenge(``\mathsf{b}_5"\land ``\mathsf{m}_5"\bigcirc ``\mathsf{g}_5")$. By following the procedure described in Section \ref{subsec:discrete plan synthesis}, we obtain the desired plans $\mathsf{plan}_1 = ((c_1,``\mathsf{r}_1")(c_3,``\mathsf{g}_1")(c_4,``\mathsf{m}_1")(c_2,``\mathsf{b}_1"))^\mathsf{\omega}$, $\mathsf{plan}_2 = (c_2,``\mathsf{b}_2")(c_4,``\mathsf{m}_2")((c_1,``\mathsf{r}_2")(c_2,``\mathsf{b}_2"))^\mathsf{\omega}$, $\mathsf{plan}_3 = (c_4,``\mathsf{m}_3")(c_3,``\mathsf{g}_3")((c_1,``\mathsf{r}_3")(c_2,``\mathsf{b}_3"))^\mathsf{\omega}$, $\mathsf{plan}_4 = ((c_3,``\mathsf{g}_4")(c_2,``\mathsf{b}_4")(c_4,``\mathsf{m}_4")(c_3,``\mathsf{g}_4"))^\mathsf{\omega}$, and $\mathsf{plan}_5 =  (c_1,``\mathsf{r}_5")((c_4,``\mathsf{m}_5")(c_3,``\mathsf{g}_5")(c_2,``\mathsf{b}_5"))^\mathsf{\omega}$. We assume that the services are provided instantly by the agents. The control gains are chosen as $\mu_{c,i} = 3$, $\mu_i = 25$,  $\mu_{\alpha,i} = 0.1$, $\forall i\in\mathcal{N}$, and $\mu_{\text{con},m} = \mu_{\text{col},m} = 0.1$, $\forall m\in\mathcal{M}_0$, $m\in\bar{\mathcal{M}}$. The simulation results are depicted in Fig. \ref{fig:gammas}-\ref{fig:a_hats}. for $t\in[0,10^3]\text{sec}$. More specifically, Fig. \ref{fig:gammas} shows the distance functions $\mathsf{md}_i\gamma_i$, $\forall i\in\mathcal{N}$. In the total time duration, all the agents execute their first goal of their respective plans, according to their assigned priorities, whereas agent $1$ executes its second goal as well; Fig. \ref{fig:beta_col} and \ref{fig:beta_con} illustrate the collision- and connectivity- associated terms $\beta_{\text{col}_m}(\iota_m)$, $\forall m\in\bar{\mathcal{M}}$, $\beta_{\text{con}_m}(\eta_{m})$, $\forall m\in\mathcal{M}_0$, which are always positive, verifying the collision avoidance and connectivity maintenance properties. Finally, Fig. \ref{fig:a_hats} depicts the adaptation variables $\hat{\alpha}_i$, $\forall i\in\mathcal{N}$, which are always kept bounded.

%%%%%%%%%%%%%%%%%%%%%%%%%%%%%%%%%%%%%%%%%%%%%%%%%%%%%%%%%%%%%%%%%%%%%%%%%%%%%%%
\section{CONCLUSIONS AND FUTURE WORKS} \label{sec:Conclusion}
This paper presented a hybrid coordination strategy for the motion planning of a multi-agent team under high level specifications expressed as LTL formulas. Inter-agent collision avoidance and connectivity maintenance is also guaranteed by the proposed continuous control protocol. Future efforts will be devoted towards addressing timed temporal tasks as well as including workspace obstacles.

%%%%%%%%%%%%%%%%%%%%%%%%%%%%%%%%%%%%%%%%%%%%%%%%%%%%%%%%%%%%%%%%%%%%%%%%%%%%%%%%
\bibliographystyle{ieeetr}
\bibliography{references}

\end{document}